\documentclass[lettersize,journal]{IEEEtran}
\usepackage{amsmath,amsfonts}
\usepackage{algorithm, algorithmicx, algpseudocode}
\usepackage{xcolor}
\usepackage{etoolbox}
\usepackage{array}
\usepackage{amsthm}
\usepackage[caption=false,font=normalsize,labelfont=sf,textfont=sf]{subfig}
\usepackage{textcomp}
\usepackage{stfloats}
\usepackage{url}
\usepackage{verbatim}
\usepackage{graphicx}
\usepackage{cite}
\usepackage{multirow}
\usepackage{graphicx}
\usepackage{anyfontsize}
\usepackage{makecell}
\usepackage{mathrsfs}
\usepackage{booktabs}
\usepackage{xcolor}
\usepackage{makecell}
\usepackage{bm}
\usepackage{etoolbox} 
\newtheorem{remark}{Remark}
\begin{document}

\title{Sensor Deployment Optimization for Passive TDOA Localization Under Unknown Drift Distribution}

\author{Zhenxing Zhang, Tianxian Zhang, \IEEEmembership{Member, IEEE},  Zerui Zhang, Zicheng Wang, Xueting Li
\thanks{Zhenxing Zhang, Tianxian Zhang,  Zerui Zhang, and Zicheng Wang  are with  the School of Electronic Engineering, University of Electronic Science and Technology of China, Chengdu 611731, China (e-mail:
	zxingzhang0512@163.com; tianxianzhang@gmail.com; 202311012208@std.uestc.edu.cn; 375756407@qq.com).   \textit{(Corresponding author: Tianxian Zhang.)}.
	
}}

\markboth{Journal of \LaTeX\ Class Files,~Vol.~14, No.~8, August~2021}%
{Shell \MakeLowercase{\textit{et al.}}: A Sample Article Using IEEEtran.cls for IEEE Journals}


\maketitle

\begin{abstract}
This paper investigates how to deploy sensors offline to provide robust passive TDOA  localization accuracy across the entire region of interest (ROI) when their positions are subject to drift errors caused by factors such as wind. Since in practice only the 1st and 2nd order statistics of sensor drift errors can be estimated from historical sensor telemetry data or wind field statistics, by using them we first derive a generalized geometric dilution of precision under drift errors ($\mathrm{GDOP_{D}}$),  which extends the traditional GDOP ($\mathrm{GDOP_{T}}$). Furthermore, we derive theoretical results related to  $\mathrm{GDOP_{D}}$ and $\mathrm{GDOP_{T}}$,  revealing that drift errors not only enlarge the value of GDOP but also reshape its distribution, thereby degrading localization performance. Then, we construct a  $\operatorname{{GDOP}_{D}}$-based min-max deployment optimization problem.  {Finally,  we propose an adaptive unidirectional particle swarm optimizer (AUPSO) to solve this challenging problem. The proposed method alleviates the premature convergence and the oscillatory behavior of the traditional PSO. Extensive simulations demonstrate the effectiveness of the proposed method.  This research provides a reliable offline sensor deployment planning framework for practical engineering scenarios, when the accurate drift error probability density function is not available.}
\end{abstract}


\section{Introduction}
\IEEEPARstart{I}{n}  modern wireless sensing systems, passive localization has been widely used in civil applications \cite{liu2018range} and electronic warfare \cite{ammar2024robust}, owing to its low power consumption and high concealability \cite{8675389}. Among various techniques, Time Difference of Arrival (TDOA)-based passive localization is particularly attractive due to its advantage of not requiring absolute time synchronization \cite{6338257}. TDOA estimates the target’s location by measuring the differences of signal arrival times across multiple distributed sensors \cite{ho1993solution}. However, studies have shown that the TDOA localization accuracy is highly sensitive to the sensor deployment \cite{phruksahiran2024analysis, cui2025spatial, meng2013optimality, lee2023pso, zheng2007accurate}. Therefore, optimizing sensor deployment at the planning stage is critical for achieving accurate TDOA localization. This work focuses on the pre-deployment design of sensors for area coverage tasks such as border surveillance or environmental monitoring. In these application scenarios, sensors are deployed in advance (e.g., UAVs are assigned to a predetermined standby formation). Although the target’s location is unknown and uncontrollable, the sensor layout can be optimized to ensure robust localization performance over the entire Region of Interest (ROI).

Recently, some studies considered sensor deployment over the entire ROI. For example,  Aubry et al. \cite{10081431} proposed a robust framework for optimizing sensor deployment over ROI. Diez et al. \cite{diez2022analysis} explored a reliable sensor under fault conditions. To address Non-Line-of-Sight (NLoS) errors, Zhao et al. \cite{zhao2022finding} developed a deployment method for Ultra-Wideband (UWB) TDOA antennas. Wei et al.\cite{wei2025sensor} considered measurement information anomalies and optimized sensor deployment to improve system tolerance. Wang et al. \cite{wang2022geometric}  focused on accurate localization by deploying sensors in multiple key subareas. {Although these studies have advanced the field, they assume position errors follow known Probability Density Functions (PDFs), typically zero-mean Gaussian, through the 1st and 2nd order statistics.}  {While analytically convenient, such assumptions may not hold in many application scenarios, such as aerial swarm monitoring, where sensor positions may drift\footnote{It may be hard to maintain sensor's position as commanded.} due to the wind \cite{lin2022efficient, jana2022cnn}.} {Chen et al. \cite{chen2018ules} introduced an evaluation scheme, ULES, that selects the most reliable beacon nodes for localization based on multi-criterion decision-making, including drift indicators. This online screening strategy effectively reduces the influence of drifted nodes during the localization process. However, such methods are inherently reactive, i.e., they compensate for drift after it has occurred, rather than proactively designing the sensors to be inherently robust to drift from the outset. When sensor deployment is optimized under idealized or oversimplified assumptions, it can lead to uncontrolled localization degradation, especially in ROIs that are originally designed for high accuracy. This highlights the need to explore proactive sensor deployment strategies that are robust to drift errors at the system design stage. Although the PDF of sensor drift errors is unknown in practice, the 1st and 2nd order statistics may be estimated from the system logs  \cite{gupta2022landing, richter1995estimating}, which can be leveraged in a proactive design.}

{Another crucial step in achieving optimal sensor deployment is to provide an efficient solution method. Usually, the GDOP-based sensor deployment optimization problem is inherently non-convex due to the nonlinear dependency on sensor-target geometry \cite{wang2022geometric}. The problem becomes even more complex when drift errors are involved. Consequently, a robust and flexible optimization method is needed. Existing methods for such problems can be broadly categorized into two paradigms: deterministic methods and meta-heuristic methods. Deterministic methods aim to find solutions through systematic and rule-based procedures. For the highly non-convex sensor deployment problem, a prevalent strategy is to reformulate the original problem into a more tractable form. This often involves convex relaxations (e.g., semidefinite relaxation) or geometric approximations to obtain a surrogate problem that can be solved with guarantees \cite{10122889, chen2011antenna, sun2014antenna}. However,  the accuracy and feasibility of the resulting solution depend critically on the validity of the introduced approximations or relaxations, which may not always hold under complex, distribution-agnostic drift error models that we consider. In contrast, meta-heuristic methods are derivative-free and guided by population-based search and stochastic operators. Their principal strength lies in directly tackling the original and unmodified problem formulation, making them particularly suitable for non-convex problems. Among these methods, the Particle Swarm Optimization (PSO) method \cite{kennedy1995particle} has gained wide adoption in sensor deployment optimization due to its conceptual simplicity and strong empirical performance \cite{ zhang2024optimal, xia2021improved, zhang2022efficient, tang2008study}.  Nevertheless, PSO suffers from two major drawbacks that undermine its optimization efficiency: 1) In each iteration, all particles are updated based on the global best position. However,  the global best position may be unreliable during early iterations,  leading to premature convergence \cite{kennedy1995particle};
2) Each particle is guided by both its personal best position and the global best position, which tends to cause oscillatory behavior.
These limitations often result in suboptimal solutions with high variance across independent trials, thereby restricting the practical applicability of PSO.}

Motivated by the above considerations, this work aims to provide a robust sensor pre-deployment strategy for TDOA localization in the ROI under drift error. The main contributions of this paper are summarized as follows:

\begin{itemize}
	\item[1)] We derive a rigorous formulation of the GDOP under sensor drift error, termed $\operatorname{GDOP_D}$, assuming only the 1st and 2nd order statistics of the drift error are available. (Section \ref{Passive TDOA with drift error})
	
	\item[2)] We establish the mathematical relationship between $\operatorname{GDOP_D}$ and the traditional GDOP ($\operatorname{GDOP_T}$), proving that $\operatorname{GDOP_D}$ generalizes $\operatorname{GDOP_T}$. Besides, we theoretically analyze the impact of drift error on localization geometry. Specifically, we prove that the drift error not only induces the non-uniform scaling of GDOP but also causes monotonic warping and asymptotic saturation of the GDOP contours. (Section \ref{Analysis of the characteristics})
	
	\item[3)] We formulate a min-max non-convex optimization problem to minimize the worst-case $\operatorname{GDOP_D}$ over the entire ROI. (Section \ref{optimization model})
	
	\item[4)] {We propose an Adaptive Unidirectional Particle Swarm Optimizer (AUPSO). Its key innovations include: i) A novel unidirectional update rule that dynamically blends global and competitively selected personal best information to suppress oscillations, and ii) An adaptive factor $\beta$ that reduces early reliance on the global best position to prevent premature convergence. These mechanisms achieve a superior exploration-exploitation balance, enabling AUPSO to reliably solve the challenging sensor deployment problem.} (Section \ref{atspso})
\end{itemize}

The rest of the paper is organized as follows: Section \ref{SYSTEM MODEL} presents the system model. Section \ref{atspso} introduces the proposed AUPSO. Simulation results are provided in Section \ref{Experimental Simulation}, and Section \ref{Conclusion} concludes the paper. This paper adopts the commonly used notations, where lowercase letters represent scalars, bold lowercase letters represent vectors, and bold uppercase letters represent matrices. Besides, $(*)^\top$, $\operatorname{tr}(*)$, $\operatorname{rank}(*)$ denote the transpose operation, matrix trace, and matrix rank, respectively.

\section{SYSTEM MODEL} \label{SYSTEM MODEL}
In this section, we first introduce the application scenario of passive TDOA  localization under sensor drift error and derive the corresponding $\operatorname{GDOP_D}$ metric. Next,  we theoretically analyze the impact of the drift errors on GDOP. Finally, a min-max non-convex sensor deployment optimization problem over an ROI is established. A list of math symbols used hereinafter is given in Table \ref{table:notation}.

\begin{table*}[htbp]
	\caption{LIST OF NOTATIONS}\label{table:notation}
	\centering
	{ \fontsize{8}{10}
		\begin{tabular}{cc||cc}
		\toprule
			\textbf{Notation} &  \textbf{Description} & \textbf{Notation} & \textbf{Description}  \\ 
			\midrule
			$M$ & Number of sensors& $c$  &  The light speed in air \\ 
			$\mathcal{A}$ &   deployment region & ${r}_{1}(\mathbf{\hat{x}}_{l},{\mathbf{s}_{1}})$ & \thead{The theoretical range-difference between \\ the first sensor and the \textit{l}th resolution} \\ 
			$\mathcal{B}$ &  ROI & 	$\mathbf{P} $  &  The matrix of localization error \\
			$\mathbf{s}_{m}$ & The commanded deployment position of the $m$th sensor & 	$\operatorname{GDOP_D}(\mathbf{\hat{x}}_l, \mathbf{S})$ & The GDOP at $\mathbf{\hat{x}}_l$  with drift error \\ 
			$\mathbf{\Delta s}_{m}^{mea}$ & The position measurement error of the $m$th sensor &	$\operatorname{GDOP_T}(\mathbf{\hat{x}}_l, \mathbf{S})$ & The traditional GDOP at $\mathbf{\hat{x}}_l$  \\ 
			$\mathbf{\Delta s}_{m}^{dri}$ & The position drift error of the $m$th sensor & $\mathbf{J}_n$  & Velocity of the $n$th particle  \\ 
			${\mathbf{\tilde{s}}_m}$ & The post-drift position of the $m$th sensor & $\mathbf{R}_n$ & Position of the $n$th particle \\
			$\mathbf{u}_{m}^{dri}$ & The  means  of the drift error of the $m$th sensor & $\mathbf{Pb}_n$  &  Personal best position of the $n$th particle \\ 
			$\bm{\Sigma}_{m}^{dri}$ & The covariance  of the drift error of the $m$th sensor  &  $w$ & The inertia weight\\
			$ \mathbf{\hat{x}}_l $ &  The position of the $l$th target & $c_1$,$c_2$ & The social acceleration coefficients \\ 
			
			$\mathbf{d\hat{x}}_{l}$  & The target localization error & $\mathbf{G}$ & The global best
			position  \\ 
			${r}_{m}(\mathbf{\hat{x}}_{l}, {\mathbf{s}_{m}})$  & \thead{ The distance between the $m$th sensor \\ and the \textit{l}th target}  & $K_{max}$ &  The maximum number
			of iterations \\ 
			${r}_{m,1}({\mathbf{\hat{x}}}_{l}, {\mathbf{s}}_{m}, {\mathbf{s}}_{1})$ & \thead{The theoretical range-difference between \\ the $m$th sensor and the first sensor at \textit{l}th target} & $N$ &  The total number of particles \\ 
			$\tau _{m,1}$ & \thead{ The theoretical time-difference between \\ the $m$th sensor and the first sensor}  & $\xi$ & Random number obeying a uniform distribution within $[0,1]$ \\ 
			$\mathrm{d} \tau _{m,1}$ & The time-difference measurement error & $\mathbf{Pb}_{\zeta}$ & Random $\mathbf{Pb}_n$ selected from the swarm \\ 
						\bottomrule
		\end{tabular}
	}
\end{table*}

\subsection{Application Scenario Description}
Without loss of generality, we consider a three-dimensional localization application scenario. A group of $M+1$ sensors is commanded to be deployed within a predefined deployment region $\mathcal{A}$ to conduct passive TDOA localization over an ROI $\mathcal{B}$.   {Each sensor determines its own position (e.g., via GPS) to enable TDOA calculation. This self-reported position (post-drift position), denoted by $\tilde{\mathbf{s}}_m$, inevitably deviates from the sensor's intended (or commanded) deployment position $\mathbf{s}_m$. We model the total deviation $\mathbf{ds}_m = \tilde{\mathbf{s}}_m - \mathbf{s}_m$ as the sum of two independent error sources\footnote{{The assumption of independence between error sources is a common modeling simplification, justified by their distinct physical origins: the measurement error $\bm{\Delta \mathbf{s}}_{m}^{mea}$ arises from the internal characteristics of the positioning system (e.g., receiver noise in GPS), whereas the drift error $\bm{\Delta \mathbf{s}}_{m}^{dri}$ is induced by external and persistent forces (e.g., wind). Notably, our core framework and the derived $\operatorname{GDOP_D}$ metric depend on the statistics of the composite error $\mathrm{d}\mathbf{s}_m$. If a future work or a specific application scenario provides evidence of correlation between $\bm{\Delta \mathbf{s}}_{m}^{mea}$ and $\bm{\Delta \mathbf{s}}_{m}^{dri}$, it can be directly incorporated by modifying the off-diagonal blocks of the combined covariance matrix, without altering the proposed methodology.}}}:
\begin{enumerate}
	\item[i)] {Position measurement error $\mathbf{\Delta s}_{m}^{mea}$: The inherent noise of the positioning system (e.g., GPS receiver noise), typically modeled as zero-mean Gaussian. The covariance of the measurement error for the $m$th sensor is \cite{10081431}
	\begin{equation} \label{position error}
		\bm{\Sigma}_m^{mea} = \sigma _{p}^2\text{diag}\left(1,1,1\right).
	\end{equation}}
	
	\item[ii)] {Drift error $\bm{\Delta \mathbf{s}}_{m}^{dri}$: A displacement caused by persistent environmental forces (e.g., wind) that prevent the sensor from perfectly maintaining its commanded deployment position $\mathbf{s}_m$. In practice, the exact PDF of $\mathbf{\Delta s}_m^{dri}$ is often unknown. However, its 1st and 2nd  statistics can be reliably estimated from historical sensor telemetry data or environmental models (e.g., wind field statistics) \cite{gupta2022landing}. We denote the empirical mean and covariance of $ \mathbf{\Delta s}_{m}^{dri}$ as  $\mathbf{u}_{m}^{dri}$ and $\bm{\Sigma}_{m}^{dri}$, respectively:
	\begin{equation}
		\mathbf{u}_{m}^{dri} = (\mu _{m,x}^{dri}, \mu _{m,y}^{dri}, \mu _{m,z}^{dri})^{\top},
	\end{equation}
	\begin{equation}
		\begin{aligned}
			&	\boldsymbol{\Sigma}{_{m}^{dri}}=\\
			& \left[ \begin{matrix}
				{{\left( \sigma _{m,x}^{dri} \right)}^{2}} & {{\rho }_{m,xy}}\sigma _{m,x}^{dri}\sigma _{m,y}^{dri} & {{\rho }_{m,xz}}\sigma _{m,x}^{dri}\sigma _{m,z}^{dri}  \\
				{{\rho }_{m,xy}}\sigma _{m,x}^{dri}\sigma _{m,y}^{dri} & {{\left( \sigma _{m,y}^{dri} \right)}^{2}} & {{\rho }_{m,yz}}\sigma _{m,y}^{dri}\sigma _{m,z}^{dri}  \\
				{{\rho }_{m,xz}}\sigma _{m,x}^{dri}\sigma _{m,z}^{dri} & {{\rho }_{m,yz}}\sigma _{m,y}^{dri}\sigma _{m,z}^{dri} & {{\left( \sigma _{m,z}^{dri} \right)}^{2}}  \\
			\end{matrix} \right],
		\end{aligned}
	\end{equation}
where $\rho_{m,xy}, \rho_{m,xz}, \rho_{m,yz}$ are the directional correlation coefficients of the drift error, allowing the model to capture the anisotropic drift errors in different directions. For conciseness, we define $\bm{\rho}_{m} = [\rho_{m,xy}, \rho_{m,xz}, \rho_{m,yz}]^\top$ and $\bm{\sigma}^{dri}_{m} = [\sigma _{m,x}^{dri}, \sigma _{m,y}^{dri},\sigma _{m,z}^{dri}]^\top$.}
\end{enumerate}

{Thus, the actual position used in TDOA calculation is $\tilde{\mathbf{s}}_m = \mathbf{s}_m + \mathbf{ds}_m$, where the composite error $\mathbf{ds}_m = \bm{\Delta \mathbf{s}}_{m}^{mea} + \bm{\Delta \mathbf{s}}_{m}^{dri}$ distorts the perceived geometry of the sensors, thereby degrading target localization accuracy.}

{Our objective is to optimize the commanded deployment positions $\mathbf{S} = [\mathbf{s}_{1},...,\mathbf{s}_{M+1}]^\top$, given the 1st and 2nd order statistics of $\mathbf{ds}_m$, to ensure robust target localization across the entire ROI $\mathcal{B}$.} Assume the $l$th target appears at $\hat{\mathbf{x}}_l = (\hat{x}_l, \hat{y}_l, \hat{z}_l)^{\top} \in \mathcal{B}$. The traditional $\operatorname{GDOP_T}(\hat{\mathbf{x}}_l, \mathbf{S})$ quantifies localization sensitivity to measurement noise and geometry, but it relies on the assumption of perfectly known sensor positions. To incorporate the effect of  $\mathbf{ds}_m$, we derive a generalized, drift-aware  $\operatorname{GDOP_D}(\hat{\mathbf{x}}_l, \mathbf{S})$ metric, which will be calculated below.

\begin{remark}
{Like the $\operatorname{GDOP_T}(\hat{\mathbf{x}}_l, \mathbf{S})$,   the derivation of the $\operatorname{GDOP_D}(\hat{\mathbf{x}}_l, \mathbf{S})$ metric relies solely on the 1st and 2nd order statistics, i.e., mean and covariance of the errors. No specific PDF  is assumed for the drift error $\mathbf{\Delta s}_m^{dri}$. This makes the subsequent optimization framework inherently distributionally robust, as it is valid for any error distribution sharing the same statistics.} 
\end{remark}

\subsection{$\operatorname{GDOP_D}( \mathbf{\hat{x}}_l, \mathbf{{S}})$ of Passive TDOA  Localization}\label{Passive TDOA with drift error}

In the ideal condition, the distance between the $m$th sensor and the ${l}$th target is
\begin{equation} \label{distance}
	{r}_{m}(\mathbf{\hat{x}}_{l}, {\mathbf{s}_{m}})={{\left\| \mathbf{\hat{x}}_l-{\mathbf{s}_{m}} \right\|}_{2}}, m = 1,2,...,M+1,
\end{equation}
where $\left\| \bullet \right\|_{2}$ is ${l_{2}}$-norm.

Without loss of generality, the first sensor $\mathbf{s}_{1}$ is selected as the reference, then the theoretical range-difference ${r}_{m,1}$ between the $m$th sensor and the reference sensor is
\begin{equation} \label{range-differe}
	\begin{aligned}
		{r}_{m,1}({\mathbf{\hat{x}}}_{l}, {\mathbf{s}}_{m}, {\mathbf{s}}_{1})&={r}_{m}(\mathbf{\hat{x}}_{l}, {\mathbf{s}_{m}})-{r}_{1}(\mathbf{\hat{x}}_{l},
		{\mathbf{s}_{1}}),\\
		&m = 2,...,M+1.
	\end{aligned}
\end{equation}
Thus, the theoretical time-difference 	$\tau _{m,1}$  is
\begin{equation} \label{cc}
	\tau _{m,1}=\frac{1}{c}{{r}_{m,1}({\mathbf{\hat{x}}}_{l}, {\mathbf{s}}_{m}, {\mathbf{s}}_{1})}		,m = 2,...,M+1,
\end{equation}
{where $c$ is the light speed in air.}

In practical TDOA localization,  the target localization error, denoted as $\mathbf{d\hat{x}}_{l} = {{\left[ \text{d}\hat{x}_{l},\text{d}\hat{y}_{l},\text{d}\hat{z}_{l} \right]}^{\top}}$, arises from time-difference measurement error\footnote{In modern TDOA networks, macroscopic long-term clock drifts are typically eliminated by synchronization protocols. Therefore, the residual $\mathrm{d}\tau_{m,1}$ primarily encompasses the nanosecond-level short-term timing jitter induced by severe mechanical vibrations. Modeling this transient jitter purely as a zero-mean noise captures these physical interferences while preserving the distribution-free nature of the analytical framework \cite{jia2022composite}.} $\mathrm{d} \tau _{m,1}$ (typically modeled as a zero-mean noise with standard deviation $\sigma _{m}^{tim}$) and the $\mathbf{{ds}}_{m}$ \cite{ho2007source}. Taking these errors into account, we can get
\begin{equation} \label{error equation2}
	\begin{aligned}
		c \text{d}{{\tau }_{m,1}}=& {r}_{m}(\mathbf{\hat{x}}_{l}+\mathbf{d\hat{x}}_{l}, {\mathbf{s}_{m}}+\mathbf{{ds}}_{m})-{r}_{m}(\mathbf{\hat{x}}_{l}, {\mathbf{s}_{m}})\\
		-&\left({r}_{1}(\mathbf{\hat{x}}_{l}+\mathbf{d\hat{x}}_{l}, {\mathbf{s}_{1}}+\mathbf{{ds}}_{1})-{r}_{1}(\mathbf{\hat{x}}_{l}, {\mathbf{s}_{1}})\right),\\
		& m=2,...,M+1.
	\end{aligned}
\end{equation}
{When $\mathbf{d\hat{x}}_{l}$, $\mathbf{{ds}}_{m}$ and $\mathbf{ds}_{1}$ are small, we perform the first-order Taylor expansion of  ${r}_{1}(\mathbf{\hat{x}}_{l}+\mathbf{d\hat{x}}_{l}, {\mathbf{s}_{1}}+\mathbf{ds}_{1})$ at $\mathbf{\hat{x}}_{l}$ and $\mathbf{s}_{1}$, and the first-order Taylor expansion of ${r}_{m}(\mathbf{\hat{x}}_{l}+\mathbf{d\hat{x}}_{l}, {\mathbf{s}_{m}}+\mathbf{{ds}}_{m})$ at $\mathbf{\hat{x}}_{l}$ and ${\mathbf{s}_{m}}$ \cite{yang2022optimal,xu2025optimal}. Then, we  get 
\begin{equation} 	\label{error equation2.1}
	\begin{aligned}
		c\mathrm{d} \tau _{m,1}  \approx &   \left(\frac{\partial {{r}_{m}(\mathbf{\hat{x}}_{l}, {\mathbf{s}_{m}})}}{\partial \mathbf{\hat{x}}_{l}}-\frac{\partial {{r}_{1}(\mathbf{\hat{x}}_{l}, {\mathbf{s}_{1}})}}{\partial \mathbf{\hat{x}}_{l}}\right)^{{\top}}\mathbf{d\hat{x}}_{l}\\
		&+ \left(\frac{\partial {{r}_{m}(\mathbf{\hat{x}}_{l}, {\mathbf{s}_{m}})}}{\partial \mathbf{s}_{m}}\right)^{\top}\mathbf{ds}_{m} - \left(\frac{\partial {{r}_{1}(\mathbf{\hat{x}}_{l}, {\mathbf{s}_{1}})}}{\partial \mathbf{s}_{1}}\right)^{\top}\mathbf{ds}_{1},\\
		&m = 2,...,M+1,
	\end{aligned}
\end{equation}
where $ \frac{\partial {{r}_{m}(\mathbf{\hat{x}}_{l}, {\mathbf{s}_{m}})}}{\partial \mathbf{\hat{x}}_{l}} = - \frac{\partial {{r}_{m}(\mathbf{\hat{x}}_{l}, {\mathbf{s}_{m}})}}{\partial \mathbf{{s}}_{m}}$.}

{Denote $\mathbf{a}_m = \frac{\partial {{r}_{m}(\mathbf{\hat{x}}_{l}, {\mathbf{s}_{m}})}}{\partial \mathbf{{s}}_{m}}$, then (\ref{error equation2.1}) becomes
\begin{equation} \label{error equation2.3}
	\begin{aligned}
		&	c\mathrm{d} \tau _{m,1}  \approx (\mathbf{a}_1^{{\top}}-\mathbf{a}_m^{{\top}})\mathbf{d\hat{x}}_{l}+ \mathbf{a}_m^{\top}\mathbf{ds}_{m} - \mathbf{a}_1^{\top}\mathbf{ds}_{1},\\
		& \quad \quad \quad \quad\quad\quad\quad\quad \quad\quad\quad \quad\quad \quad\quad m = 2,...,M+1.
	\end{aligned}
\end{equation}
Writing (\ref{error equation2.3}) in matrix form, we have
\begin{equation}
	\mathbf{dv}=\mathbf{C}\mathbf{d\hat{x}}_{l}+\mathbf{dx_s}, 
\end{equation}
where 
\begin{equation} \label{CCC}
	\mathbf{C}=	{{\left[ \begin{matrix}
				\mathbf{a}_1^{{\top}}-\mathbf{a}_2^{{\top}} \\
				\vdots    \\
				\mathbf{a}_1^{{\top}}-\mathbf{a}_{M+1}^{{\top}}   \\
			\end{matrix} \right]}_{M\times3}},
\end{equation}
\begin{equation} \label{CCC1}
	\mathbf{dv}=c{{\left[ \text{d}\tau _{2,1}^{{}},...,\text{d}\tau _{M+1,1}^{{}} \right]}^{\top}},
\end{equation}
\begin{equation} \label{CCC2}
	\mathbf{dx_s}= {{\left[ \begin{matrix}
				\mathbf{a}_2^{\top}\mathbf{ds}_{2}- \mathbf{a}_1^{\top}\mathbf{ds}_{1} \\
				\vdots    \\
				\mathbf{a}_{M+1}^{\top}\mathbf{ds}_{M+1} - \mathbf{a}_1^{\top}\mathbf{ds}_{1}   \\
			\end{matrix} \right]}_{M\times3}}.
\end{equation}}

Choose suitable sensors so as to make $\operatorname{rank}(\textbf{C})=3$, and then the localization error can be obtained as 
\begin{equation}
	\mathbf{d\hat{x}}_{l}={{\left( {{\mathbf{C}}^{\top}}\mathbf{C} \right)}^{-1}}{{\mathbf{C}}^{\top}}\left[ \mathbf{dv}-\mathbf{dx_s} \right].
\end{equation}
To facilitate derivation, let $\textbf{D}$ be $({\mathbf{C}}^{\top}\mathbf{C})^{-1}{\mathbf{C}}^{\top}$. {Besides, the time-difference measurement error, position measurement error, and drift error are assumed mutually independent—a common and reasonable simplification given their distinct physical origins (signal processing, navigation system, and environmental forces, respectively)}. Under this assumption, the matrix of localization error is calculated as 
\begin{equation}
	\begin{aligned}
		\mathbf{P} & \stackrel{\Delta}{=}  E\left( 	\mathbf{d\hat{x}}_{l}{{\left( 	\mathbf{d\hat{x}}_{l} \right)}^{\top}} \right)\\
		&=\mathbf{D}\left( E\left( \mathbf{dv} {{\left( \mathbf{dv} \right)}^{\top}} \right)+E\left( \mathbf{dx_s} {{\left( \mathbf{dx_s} \right)}^{\top}} \right) \right)\mathbf{D}^\top.
	\end{aligned}
\end{equation}
{The $\operatorname{GDOP_D}(\mathbf{\hat{x}}_l, \mathbf{S})$ is 
\begin{equation} \label{gdop_t}
	\operatorname{GDOP_D}(\mathbf{\hat{x}}_l, \mathbf{S})
	=\sqrt{\operatorname{tr}\left(\mathbf{P}  \right)}.
\end{equation}}

If all measurements are gathered and processed at the fusion center,  all time-difference measurement errors are related to each other at the fusion center. Assuming that the standard deviation of each time-difference measurement error is the same, i.e., $\sigma _{m}^{tim} = \sigma _{\tau}, m=2,..., M+1$, and that any pair of time-difference measurement errors share a common correlation coefficient $\eta$ of TDOA measurement errors  \cite{xie2018analysis}. Then, $E( \mathbf{dv}(\mathbf{dv})^{\top})$ is
\begin{equation}
	E( \mathbf{dv}(\mathbf{dv})^{\top}) = 	({c}\sigma_{\tau})^2\left((1-\eta)\mathbf{I}_{M\times M}+\eta \mathbf{L}_{M\times M}\right).
\end{equation}
where $\mathbf{I}_{M\times M}$ is the $M \times M$ identity matrix, $\mathbf{L}_{M\times M}$ is the ${M \times M}$ square matrix whose elements are all 1.

{Besides, $E\left( \mathbf{dX_s}{{\left( \mathbf{dX_s} \right)}^{\top}} \right)$ is
\begin{equation} \label{dX}
	E\left( \mathbf{dX_s}{{\left( \mathbf{dX_s} \right)}^{\top}} \right) = 	\sigma _{p}^{2}\left( {{\mathbf{I}_{M\times M}}}+{{\mathbf{L}_{M\times M}}} \right)  + \mathbf{F_1}+\mathbf{F}_2,
\end{equation}
where
\begin{equation}
	\begin{aligned}
		&{{\mathbf{F}}}_1=\left(\mathbf{a}_{1}^{\top}\boldsymbol{\Sigma}_{1}^{dri}\mathbf{a}_{1}\right)\mathbf{L}_{M\times M}+ \boldsymbol{\Lambda}_{M\times M},\\
		&	{\mathbf{F}}_2=\mathbf{b}\mathbf{b}^{\top}.
	\end{aligned}
\end{equation}
Here
\begin{equation}
	\begin{aligned}
		&	\boldsymbol{\Lambda} = \text{diag}\left( \mathbf{a}_{2}^{\top}\boldsymbol{\Sigma}{_{2}^{dri}}\mathbf{a}_{2},...,\mathbf{a}_{M+1}^{\top}\boldsymbol{\Sigma}{_{M+1}^{dri}}\mathbf{a}_{M+1} \right),\\
		&\mathbf{b} = (b_1,...,b_{M})^{\top},\\
		&	b_m  = \mathbf{a}_{m+1}^{\top}\mathbf{u}_{m+1}^{dri}-\mathbf{a}_{1}^{\top}\mathbf{u}_{1}^{dri}, m=1,...,M.
	\end{aligned}
\end{equation}
The detailed derivation is shown in Appendix.}

{For convenience, let ${{\mathbf{H}_{\tau}}} = E( \mathbf{dv}(\mathbf{dv})^{\top})$, ${{\mathbf{H}_{p}}} =\sigma _{p}^{2}\left( {{\mathbf{I}_{M\times M}}}+{{\mathbf{L}_{M\times M}}} \right)$, $\mathbf{P}$ is  then calculated as 
\begin{equation}
	\mathbf{D}\left( {\mathbf{H}}_{\tau}+{\mathbf{H}}_{p}+{\mathbf{F}}_1+{\mathbf{F}}_{2} \right){{\mathbf{D}}^{\top}}.
\end{equation}
Thus, $\operatorname{GDOP_D}(\mathbf{\hat{x}}_l, \mathbf{S})$ is  
\begin{equation} \label{gdop_t}
	\operatorname{GDOP_D}(\mathbf{\hat{x}}_l, \mathbf{S}) =\sqrt{\operatorname{tr}\left( \mathbf{D}\left( {{\mathbf{H}}_{\tau}}+{\mathbf{H}}_{p}+{\mathbf{F}}_1+{\mathbf{F}}_{2} \right){{\mathbf{D}}^{\top}} \right)}.
\end{equation}
Then, by some algebra, we obtain
\begin{equation} \label{tr1}
	\operatorname{tr}\left( \mathbf{D} {{\mathbf{H}}_{\tau}}\mathbf{D}^\top\right) = (c\sigma_{\tau})^2 \left(\left(1-\eta\right)\operatorname{tr}\left( \mathbf{D}\mathbf{D}^\top\right)+ \eta \operatorname{tr}\left( \mathbf{D} {{\mathbf{L}}}\mathbf{D}^\top\right) \right),\\
\end{equation}
\begin{equation} \label{tr2}
	\operatorname{tr}\left( \mathbf{D} {{\mathbf{H}}_{p}}\mathbf{D}^\top\right)= \sigma_{p}^{2} \left(\operatorname{tr}\left( \mathbf{D}\mathbf{D}^\top\right)+ \operatorname{tr}\left( \mathbf{D} {{\mathbf{L}}}\mathbf{D}^\top\right) \right),\\
\end{equation}
\begin{equation} \label{tr3}
	\operatorname{tr}\left( \mathbf{D} {{\mathbf{F}}_{1}}\mathbf{D}^\top\right)=  \left(\mathbf{a}_{1}^{\top}\boldsymbol{\Sigma}_{1}^{dri}\mathbf{a}_{1}\right)\operatorname{tr}\left( \mathbf{D} {{\mathbf{L}}}\mathbf{D}^\top\right)
	+ \operatorname{tr}\left( \mathbf{D} \boldsymbol{\Lambda}\mathbf{D}^\top\right),
\end{equation}
\begin{equation} \label{tr4}
	\operatorname{tr}\left( \mathbf{D} {{\mathbf{F}}_{2}}\mathbf{D}^\top\right)= \operatorname{tr}\left( \mathbf{D} \mathbf{b}\mathbf{b}^\top\mathbf{D}^\top\right).
\end{equation}
Thus, we can get the following proposition.}

\textbf{Proposition:}
{\begin{equation} \label{Proposition2}
	\begin{aligned}
		&\operatorname{GDOP^2_D}(\mathbf{\hat{x}}_l, \mathbf{S})\\
		\quad\quad&=	\left(\left(1-\eta\right)(c\sigma_{\tau})^2+\sigma_{p}^{2}\right)\operatorname{tr}\left( \mathbf{D}\mathbf{D}^\top\right) \\ 
		& \quad	+\left(\eta(c\sigma_{\tau})^2+\sigma_{p}^{2}+\mathbf{a}_{1}^{\top}\boldsymbol{\Sigma}{_{1}^{dri}}\mathbf{a}_{1}\right)\operatorname{tr}\left( \mathbf{D}{{\mathbf{L}}}\mathbf{D}^\top\right)\\
		&\quad+  \operatorname{tr}\left( \mathbf{D} \boldsymbol{\Lambda}\mathbf{D}^\top\right)+	\operatorname{tr}\left( \mathbf{D}\mathbf{b}\mathbf{b}^{\top}\mathbf{D}^{\top}\right). \\ 
	\end{aligned}
\end{equation}
The above proposition  reveals a mathematical coupling effect between the time-domain measurement noise and the spatial geometric topology. The  noise $(c\sigma_\tau)^2$ and the spatial drift penalties (e.g., $\mathbf{a}_{1}^{\top}\boldsymbol{\Sigma}{_{1}^{dri}}\mathbf{a}_{1}$) act as multipliers for the geometry trace matrices $\operatorname{tr}\left( \mathbf{D}\mathbf{D}^\top\right)$ and $\operatorname{tr}\left( \mathbf{D}{{\mathbf{L}}}\mathbf{D}^\top\right)$.  If the nominal sensor deployment $\mathbf{S}$ is geometrically fragile under wind-induced drift, these trace matrices will become ill-conditioned and spike exponentially, catastrophically amplifying both the vibration noise and drift penalties. This highlights the indispensable role of spatial optimization: by locking down the baseline geometry ($\mathbf{D}$) against unpredictable drifts, we can prevent severe transient measurement noise from being destructively amplified. In the following subsection, we provide a detailed discussion on how drift error affects the GDOP.

\subsection{Analysis of the Characteristics of $\operatorname{GDOP_{D}}(\mathbf{\hat{x}}_l, \mathbf{S})$} \label{Analysis of the characteristics}
{While it is intuitive that drift error degrades performance, the key questions are: How exactly does it distort the localization geometry? And what are the actionable insights for system design? The following theoretical analysis answers these questions by deriving exact closed-form relationships. Crucially, we reveal non-intuitive phenomena—such as the high drift-error sensitivity of high-accuracy regions and the warping of error contours—that fundamentally motivate our drift-aware optimization formulation and method design.}

{For the convenience of theoretical analysis,  let $\boldsymbol{\bm{\rho}}_{m} = [0, 0,0]^\top$, $\boldsymbol{\bm{\sigma}}^{dri}_{m} = \sigma_{d} \mathbf{e}$, where  $\mathbf{e} = (1,1,1)^\top$, for all $m = 1,...,M+1$. Then, (\ref{tr3})  becomes
\begin{equation} \label{tr3-1}
	\operatorname{tr}\left( \mathbf{D} {{\mathbf{F}}_{1}}\mathbf{D}^\top\right)=  \sigma^2_{d}\operatorname{tr}\left( \mathbf{D} {{\mathbf{L}}}\mathbf{D}^\top\right)
	+ \sigma^2_{d}\operatorname{tr}\left( \mathbf{D} \mathbf{D}^\top\right).
\end{equation}
Let $\mathbf{u}_{m}^{dri} = \mu_d\mathbf{e}$ for all $m = 1,...,M+1$, then ${\mathbf{F}}_{2}$ = $\mu^2_d\mathbf{C}\mathbf{e}\mathbf{e}^\top\mathbf{C}^\top$, thus (\ref{tr4}) becomes
\begin{equation} \label{tr4-1}
	\begin{aligned}
		\operatorname{tr}\left( \mathbf{D} {{\mathbf{F}}_{2}}\mathbf{D}^\top\right) = \mu _{d}^{2}\operatorname{tr}\left(  \mathbf{e}\mathbf{e}^{\top}\right) = 3 \mu _{d}^{2} .
	\end{aligned}
\end{equation}}

{Bringing (\ref{tr3-1}) and (\ref{tr4-1}) into (\ref{Proposition2}), we can get
\begin{equation} \label{Proposition2-2}
	\begin{aligned}
		\operatorname{GDOP^2_D}(\mathbf{\hat{x}}_l, \mathbf{S}) =&	\left(\left(1-\eta\right)(c\sigma_{\tau})^2+\sigma_{p}^{2}+\sigma _{d}^{2}\right)\operatorname{tr}\left( \mathbf{D}\mathbf{D}^\top\right) \\ 
		& +\left(\eta(c\sigma_{\tau})^2+\sigma_{p}^{2}+\sigma _{d}^{2}\right)\operatorname{tr}\left( \mathbf{D}{{\mathbf{L}}}\mathbf{D}^\top\right)+ 3\mu^2 _{d}. \\ 
	\end{aligned}
\end{equation}}

\newtheorem{thm}{\bf Corollary}
\begin{thm}\label{thm1} {The drift-aware  $\operatorname{GDOP_D}(\mathbf{\hat{x}}_l, \mathbf{S})$ generalizes the traditional $\operatorname{GDOP_T}(\mathbf{\hat{x}}_l, \mathbf{S})$ by explicitly incorporating the drift error. Consequently, the presence of drift error systematically degrades the geometric precision, i.e., $\operatorname{GDOP_D}(\mathbf{\hat{x}}_l, \mathbf{S}) > \operatorname{GDOP_T}(\mathbf{\hat{x}}_l, \mathbf{S})$ whenever the drift error is non-zero ($\sigma_d \neq 0$ or $\mu_d \neq 0$).}
\end{thm}
\begin{IEEEproof}[Proof]
	When drift error is ignored, and only time-difference measurement and position measurement errors are considered,
	the $\operatorname{GDOP_{T}}(\mathbf{\hat{x}}_l, \mathbf{S})$ is
	\begin{equation} \label{GDOPT}
		\begin{aligned}
			\operatorname{GDOP^2_T}(\mathbf{\hat{x}}_l, \mathbf{S}) =&	\left(\left(1-\eta\right)(c\sigma_{\tau})^2+\sigma_{p}^{2}\right)\operatorname{tr}\left( \mathbf{D}\mathbf{D}^\top\right) \\ 
			& +\left(\eta(c\sigma_{\tau})^2+\sigma_{p}^{2}\right)\operatorname{tr}\left( \mathbf{D}{{\mathbf{L}}}\mathbf{D}^\top\right). \\ 
		\end{aligned}
	\end{equation}
	Combining (\ref{Proposition2-2}) and (\ref{GDOPT}) yields
	\begin{equation} \label{Proposition2-3}
		\begin{aligned}
			\operatorname{GDOP^2_D}(\mathbf{\hat{x}}_l, \mathbf{S})=&\operatorname{GDOP^2_T}(\mathbf{\hat{x}}_l, \mathbf{S})\\
			&+\sigma _{d}^{2}\left(\operatorname{tr}\left( \mathbf{D}\mathbf{D}^\top\right)+\operatorname{tr}\left( \mathbf{D}\mathbf{L}\mathbf{D}^\top\right)\right)+ 3\mu^2 _{d}.
		\end{aligned}
	\end{equation}
{	When $\operatorname{rank}(\textbf{C})=3$, the matrix \textbf{C} in (\ref{CCC}) is decomposed by singular value decomposition (SVD):
	\begin{equation} \label{CSVD}
		\mathbf{C} = \mathbf{V}_{C}\boldsymbol{\Sigma}_{C}\mathbf{U} _{C}^{\top},
	\end{equation}
	where $\mathbf{V}_{C}$ and $\mathbf{U}_{C}$ are $M\times M$ and $3\times 3$  unitary matrices, respectively, $\boldsymbol{\Sigma}_{C}$ is a flat $M\times 3$ diagonal matrix whose diagonal nonzero elements are  ${\lambda}_{i}>0, i =1,2,3$. Thus,  we can get
	\begin{equation} \label{DD}
		\begin{aligned}
			\operatorname{tr}\left( \mathbf{D}\mathbf{D}^\top\right) &=\operatorname{tr}\left({\left( {{\mathbf{C}}^{\top}}\mathbf{C} \right)}^{-1}\right)\\
			&	=\operatorname{tr}\left({{\mathbf{U}}_{C}} \left( \boldsymbol{\Sigma}_{C}^{\top} \boldsymbol{\Sigma}_{C} \right)^{-2}\mathbf{U}_{C}^{\top}\right)\\
			&=\sum^3_{i=1}{\lambda}^{-2}_{i} >0.
		\end{aligned}
	\end{equation} 
	Besides, we know that $\mathbf{L} = \mathbf{e}\mathbf{e}^\top$, where $\mathbf{e} = (1,...,1)^\top \in \mathbb{R}^M$. Thus, we can get
	\begin{equation} \label{DDds}
		\operatorname{tr}\left( \mathbf{D}\mathbf{L}\mathbf{D}^\top\right) = \operatorname{tr}\left( \mathbf{D}\mathbf{e}(\mathbf{D}\mathbf{e})^\top\right) =\left\| \mathbf{D}\mathbf{e} \right\|^2_{2}> 0.
	\end{equation} }
	
{	Bringing (\ref{DD}) and (\ref{DDds}) into (\ref{Proposition2-3}), we get $\operatorname{GDOP^2_D}(\mathbf{\hat{x}}_l, \mathbf{S})>\operatorname{GDOP^2_T}(\mathbf{\hat{x}}_l, \mathbf{S})$, when $\sigma _{d},\mu_d \neq 0$.}
\end{IEEEproof}

{System implication: The inequality in Corollary 1 provides a fundamental justification for our work: the sensor deployment optimized solely using the traditional $\operatorname{GDOP_T}$ metric may become inherently suboptimal when drift error is present, as it underestimates the true localization error. The degree of suboptimality is connected with the drift error, position measurement error, and the time-difference measurement error. Therefore, $\operatorname{GDOP_D}$ is not merely an alternative metric but a necessary one for robust system design under realistic conditions where drift errors cannot be neglected. }

\begin{thm}\label{thm2} Linear-scale transformation of GDOP under zero‑Mean drift error with $\mu _{d} = 0$, $\eta = 0.5$. If $\mu _{d} = 0$ and  $\eta = 0.5$, 
	$\operatorname{GDOP_{D}}(\mathbf{\hat{x}}_l, \mathbf{S})$ and  $\operatorname{GDOP_{T}}(\mathbf{\hat{x}}_l, \mathbf{S})$ exhibit a strict linear affine relationship with a fixed scale ratio  $\aleph$,
	\begin{equation}
		\aleph=	\sqrt{\frac{0.5(c\sigma_{\tau})^2+\sigma_{p}^{2}+\sigma _{d}^{2}}{0.5{{c}^{2}}\sigma_{\tau}^2+\sigma_{p}^{2}}}.
	\end{equation}
\end{thm}
\begin{IEEEproof}[Proof]
	Bring $\eta = 0.5$, $\mu _{d} = 0$ into (\ref{Proposition2-2}) and (\ref{GDOPT}), we obtain
	\begin{equation} \label{cpllory2-1}
		\begin{aligned}
			& \operatorname{GDOP^2_{T}}(\mathbf{\hat{x}}_l, \mathbf{S}) \\
			&=  \left(0.5(c\sigma_{\tau})^2+\sigma_{p}^{2}\right) \left(\operatorname{tr}\left( \mathbf{D}\mathbf{D}^\top\right)+ \operatorname{tr}\left( \mathbf{D} {{\mathbf{L}}}\mathbf{D}^\top\right) \right),
		\end{aligned}
	\end{equation}
	and 
	\begin{equation} \label{cpllory2-2}
		\begin{aligned}
			& \operatorname{GDOP^2_{D}}(\mathbf{\hat{x}}_l, \mathbf{S}) \\
			&=  \left(0.5(c\sigma_{\tau})^2+\sigma_{p}^{2}+\sigma _{d}^{2}\right) \left(\operatorname{tr}\left( \mathbf{D}\mathbf{D}^\top\right)+ \operatorname{tr}\left( \mathbf{D} {{\mathbf{L}}}\mathbf{D}^\top\right) \right).
		\end{aligned}
	\end{equation}
	Combining (\ref{cpllory2-1}) and (\ref{cpllory2-2}), we yield that
	\begin{equation}
		\frac{\operatorname{GDOP_{D}}(\mathbf{\hat{x}}_l, \mathbf{S})}{\operatorname{GDOP_{T}}(\mathbf{\hat{x}}_l, \mathbf{S})} = \sqrt{\frac{0.5(c\sigma_{\tau})^2+\sigma_{p}^{2}+\sigma _{d}^{2}}{0.5(c\sigma_{\tau})^2+\sigma_{p}^{2}}}=\aleph.
	\end{equation}
\end{IEEEproof}
{System implication: The result in Corollary 2 reveals a critical boundary case. Under these specific conditions, the spatial distribution of GDOP  remains unchanged. Consequently, a deployment optimized for $\operatorname{GDOP_T}$ is also optimal for $\operatorname{GDOP_D}$, merely suffering a uniform performance penalty. This corollary sharply delineates the conditions under which drift-agnostic approaches may remain valid, and underscores that any deviation from $\mu_d=0$ or $\eta=0.5$ breaks this invariance, necessitating our generalized metric.}
\begin{thm}\label{thm3} Mean‑driven non‑uniform scaling of GDOP with $\eta = 0.5$. If $\eta = 0.5$, the ratio
	$\operatorname{GDOP_{D}}(\mathbf{\hat{x}}_l, \mathbf{S})/ \operatorname{GDOP_{T}}(\mathbf{\hat{x}}_l, \mathbf{S})$ decreases monotonically  as $ \operatorname{GDOP_{T}}(\mathbf{\hat{x}}_l, \mathbf{S})$ increases.
\end{thm}
\begin{IEEEproof}[Proof]
	Bringing  $\eta = 0.5$ into (\ref{Proposition2-2}) and (\ref{GDOPT}), we  get
	\begin{equation} \label{cpllory3-1}
		\begin{aligned}
			& \operatorname{{GDOP}^2_{T}}(\mathbf{\hat{x}}_l, \mathbf{S}) \\
			&=  \left(0.5(c\sigma_{\tau})^2+\sigma_{p}^{2}\right) \left(\operatorname{tr}\left( \mathbf{D}\mathbf{D}^\top\right)+ \operatorname{tr}\left( \mathbf{D} {{\mathbf{L}}}\mathbf{D}^\top\right) \right),
		\end{aligned}
	\end{equation}
	and 
	\begin{equation} \label{cpllory3-2}
		\begin{aligned}
			&  \operatorname{{GDOP}^2_{D}}(\mathbf{\hat{x}}_l, \mathbf{S}) \\
			&=  \left(0.5(c\sigma_{\tau})^2+\sigma_{p}^{2}+\sigma _{d}^{2}\right) \left(\operatorname{tr}\left( \mathbf{D}\mathbf{D}^\top\right)+ \operatorname{tr}\left( \mathbf{D} {{\mathbf{L}}}\mathbf{D}^\top\right) \right)\\
			& + 3\mu^2 _{d}.
		\end{aligned}
	\end{equation}
	Combining (\ref{cpllory3-1}) and (\ref{cpllory3-2}), we obtain
	\begin{equation}
		\frac{\operatorname{GDOP_D}(\mathbf{\hat{x}}_l, \mathbf{S})}{\operatorname{GDOP_{T}}(\mathbf{\hat{x}}_l, \mathbf{S})} = \sqrt{\aleph^2 + \frac{3\mu^2 _{d}}{\operatorname{GDOP^2_{T}}(\mathbf{\hat{x}}_l, \mathbf{S})}}.
	\end{equation}
	
	Corollary 3 implies that
	\begin{itemize}
		\item[$\bullet$] Performance in low-$\operatorname{{GDOP}_{T}}(\mathbf{\hat{x}}_l, \mathbf{S})$ regions deteriorates more severely.
		\item[$\bullet$] The ratio in large-$\operatorname{{GDOP}_{T}}(\mathbf{\hat{x}}_l, \mathbf{S})$ regions tends toward a constant $\aleph$, especially when $3\mu^2 _{d} \ll \operatorname{{GDOP}^2_{T}}(\mathbf{\hat{x}}_l, \mathbf{S})$.
	\end{itemize}
\end{IEEEproof}
{System implication: The result in Corollary 3 is a counterintuitive and pivotal finding. It states that areas originally designed for the highest localization accuracy (low $\operatorname{GDOP_{T}}$) experience the most severe relative degradation under biased drift. This may undermine a common design philosophy of optimizing for average or peak accuracy. It provides a direct mathematical rationale for our min-max optimization objective (Section II-D): to ensure robust coverage, we must explicitly guard against the worst-case performance.}

Collectively, these theoretical results move beyond the intuitive notion that ‘drift error hurts localization performance.’ They quantify how drift distorts the error geometry—through non-uniform scaling, contour warping, and saturation—providing the precise mathematical foundation that necessitates and guides the formulation of our drift-robust deployment optimization problem.

\subsection{Sensor Deployment Optimization Problem}\label{optimization model}
In this section, we present the sensor deployment optimization model. To enable the sensors to effectively serve the ROI $\mathcal{B}$, we optimize their positions to minimize the maximum GDOP within the ROI. Specifically, the optimization problem is modeled as
\begin{equation} \label{model-1}
	\begin{aligned}
		&\quad \quad \quad \quad \min_{\mathbf{S}} \operatorname{GDOP^{max}_D} \\
		&\operatorname{GDOP^{max}_D} = \max_{\forall \mathbf{\hat{x}}_l \in \mathcal{B}} \operatorname{GDOP_D}(\mathbf{\hat{x}}_l, \mathbf{S})\\
		&	\quad \text{s.t.} \quad \operatorname{rank}(\mathbf{C}) = 3, \\
		&	\quad \quad \quad\mathbf{s}_m \in \mathcal{A}, \\
		& \quad\quad \quad	\forall m \in \{1, \ldots, M+1\}.
	\end{aligned}
\end{equation}
While the ROI is a bounded set, evaluating the objective $\text{GDOP}^{\max}_{\mathrm{D}}$ in (\ref{model-1}) requires finding the maximum over the entire continuous set $\mathcal{B}$. This poses a computational challenge. To render the problem tractable, we approximate the continuous ROI $\mathcal{B}$ by a finite set of $L$ sample points. This yields the problem (\ref{model-2}), which has a finite number of  operations:
\begin{equation} \label{model-2}
	\begin{aligned}
		&\quad \quad \quad \quad \min_{\mathbf{S}} \operatorname{GDOP^{max}_D} \\
		&\operatorname{GDOP^{max}_D} = \max_{l \in \{1, \ldots, L\}} \operatorname{GDOP_D}(\mathbf{\hat{x}}_l, \mathbf{S}) \\
		&\quad	\text{s.t.} \quad \operatorname{rank}(\mathbf{C}) = 3, \\
		&\quad	\quad \quad \mathbf{s}_m \in \mathcal{A}, \\
		&\quad	\quad\quad \forall m \in \{1, \ldots, M+1\}.
	\end{aligned}
\end{equation}

\section{{AUPSO: Solving Sensor deployment Optimization Problem}}\label{atspso}
In this section, we introduce the proposed AUPSO to solve the (\ref{model-2}). 
\subsection{The Shortcomings of Traditional PSO}
The traditional particle swarm optimization (TPSO) method \cite{kennedy1995particle} employs a single velocity update mechanism, in which each particle adjusts its velocity and position based on both its personal best position and the global best position. Let velocity, position, and personal best position of the $n$th particle be denoted by $\mathbf{J}_n$, $\mathbf{R}_{n}$, and $\mathbf{Pb}_n$, respectively.  The velocity and position update rules of the 
$n$th particle at the $k$th iteration can be specifically described by
\begin{equation}
	\begin{aligned} \label{39}
		\mathbf{J}^{k+1}_{n} =& w  \mathbf{J}^{k}_{n}+c_1 \xi_1 \left(\mathbf{Pb}^k_n-\mathbf{R}^k_{n}\right) +c_2 \xi_2  \left(\mathbf{G}^k-\mathbf{R}^k_{n}\right),\\
		& \quad \quad \quad \mathbf{R}^{k+1}_{n} = \mathbf{R}^{k}_{n} +\mathbf{J}^{k+1}_{n},		
	\end{aligned}
\end{equation}
where $w$ represents the inertia weight, $c_1$ and $c_2$ are the cognitive and social acceleration coefficients, respectively. $\xi_1$ and $\xi_2$ are two random numbers obeying a uniform distribution within the range $[0,1]$. $\mathbf{G}$ is the global best position identified by the entire swarm up to the $k$th iteration. 

Nevertheless, the TPSO suffers from two key limitations that affect its global search capability: 1) In each iteration, all particles are updated based on the global best position. However,  the global best position may be unreliable during early iterations,  leading to premature convergence \cite{kennedy1995particle}.
2) Each particle is guided by both its personal best position and the global best position, which tends to cause oscillatory behavior.
These limitations often result in suboptimal solutions with high variance across independent trials, thereby restricting the practical applicability of PSO.

\subsection{AUPSO}
To address these issues, we propose AUPSO. The detailed procedure is outlined below:

\textit{\textbf{Step 1. Initialization Stage:}} Let the maximum number of iterations be denoted by $K_{max}$, and initialize the current iteration index to $k=1$. The total number of particles is set as $N$. The position of the \textit{n}th particle is $\mathbf{R}_n = (\mathbf{{s}}_1,...,\mathbf{{s}}_{M+1})^\top$, which represents a vector of candidate sensor deployment scheme.   In this stage, all particle positions and velocities are initialized randomly within the deployment region by Latin hypercube sampling (LHS) \cite{loh1996latin}. 

\textit{\textbf{Step 2. Iteration Stage:}}
During this stage, particles iteratively search for the optimal sensor deployment scheme. This stage mainly involves two important substeps: velocity and position update, and personal/global best position update.

\textit{\textbf{Substep 2.1. Position and Velocity Update:}}
For the $n$th particle,  the velocity and position updates are given by
\begin{equation}\label{subswarms1}
	\begin{aligned}
		&\mathbf{J}^{k+1}_{n} = w  \mathbf{J}^{k}_{n}+c_1 \xi_1 \left(\beta\mathbf{Q}^k_{n}+(1-\beta)\mathbf{G}^k-\mathbf{R}^k_{n}\right),\\
		& \quad \quad	\quad\quad\quad\quad\mathbf{R}^{k+1}_{n} = \mathbf{R}^{k}_{n} +\mathbf{J}^{k+1}_{n},\\
	\end{aligned}
\end{equation}
where $\beta = 1-k/K_{max}$. $\mathbf{Q}^{k}_{n}$ are defined as
\begin{equation} \label{Q}
	\mathbf{Q}^{k}_{n}=\left\{ \begin{matrix}
		\mathbf{Pb}^k_{\zeta_1}, \textbf{if} \: \operatorname{GDOP^{max}_D}  \left( \mathbf{Pb}^k_{\zeta_1} \right)\le \operatorname{GDOP^{max}_D}  \left(  \mathbf{Pb}^k_{\zeta_2}  \right), \\
		\mathbf{Pb}^k_{\zeta_2}, \textbf{if} \: \operatorname{GDOP^{max}_D} \left( \mathbf{Pb}^k_{\zeta_1} \right)>\operatorname{GDOP^{max}_D}\left( \mathbf{Pb}^k_{\zeta_2} \right),  \\
	\end{matrix} \right.
\end{equation}
where $\mathbf{Pb}^k_{\zeta_1}$ and $\mathbf{Pb}^k_{\zeta_2}$ are randomly selected from the swarm. 

For each sensor $m \in {1,...,M+1}$, 
\begin{equation}\label{QQ}
	\begin{aligned}
		\text{if} \quad  {\mathbf{{s}}_{m}} \notin \mathcal{A}, \text{then,  re-sample} \quad  \mathbf{{s}}_{m} \quad  \text{from} \quad  \mathcal{A} \quad \text{such that} \\
		\mathbf{{s}}_{m} \in \mathcal{A}, \forall m \in {1,...,M+1}.
	\end{aligned}
\end{equation}

\textit{\textbf{Substep 2.2. Update of Personal Best Position and Global Best Position:}} 
For each particle, the personal best position $\mathbf{Pb}^k_{n}$ is updated  as follows:
\begin{equation} \label{personal}
	\mathbf{Pb}_{n}^{k+1}=\left\{ \begin{matrix}
		\mathbf{R}_{n}^{k+1},\textbf{if} \: \operatorname{GDOP^{max}_D} \left( \mathbf{R}_{n}^{k+1} \right)\le \operatorname{GDOP^{max}_D} \left( \mathbf{Pb}_{n}^{k} \right), \\
		\mathbf{Pb}_{n}^{k}, \textbf{if} \: \operatorname{GDOP^{max}_D} \left( \mathbf{R}_{n}^{k+1} \right)>\operatorname{GDOP^{max}_D}\left( \mathbf{Pb}_{n}^{k} \right).  \\
	\end{matrix} \right.
\end{equation}
The global best position $\mathbf{G}^{k+1}$  is then updated by selecting the best one from all the personal best positions, described by
\begin{equation} \label{G}
	{\mathbf{G}^{k+1}}=\arg \min \left\{ \operatorname{GDOP^{max}_D} \left( \mathbf{Pb}_{n}^{k+1} \right),n=1,...,N \right\}.
\end{equation}

\textit{\textbf{Step 3. Termination Check:}} 
The iteration process terminates when $k > K_{max}$, and meanwhile, the global best position $\mathbf{G}$ is output as the optimal deployment scheme. Otherwise, the iteration continues with $k=k+1$.

The Pseudocode of AUPSO is shown in Algorithm. \ref{ATSPSO}.

\begin{algorithm}[!t]
	\renewcommand{\algorithmicrequire}{\textbf{Input:}}
	\renewcommand{\algorithmicensure}{\textbf{Output:}}
	\caption{{\textbf{:} Pseudocode of AUPSO}}
	\label{ATSPSO}
	\begin{algorithmic}[1]
		\Require Optimized problem (\ref{model-2}); deployment region (search space) $\mathcal{A}$; the maximum number of iterations  $K_{max}$; 
		\Ensure  $\mathbf{G}$; 	\\
		/*  Initialization */
		\State 	$k=1$; 
		\For{1 $\leq$ \textit{n} $\leq$ \textit{N}}
		\State Randomly initialize $\mathbf{R}^1_n$ and $\mathbf{J}^1_n$, and evaluate $\operatorname{GDOP^{max}_D} (\mathbf{R}^1_n)$;
		\State Set $\mathbf{Pb}^1_{n}={\mathbf{R}^1_n}$ and update $\mathbf{G}^{1}$ by (\ref{G});
		\EndFor
		\State 	$k=k+1$; \\
		/* Main Loop */
		\While{$k \leq K_{max}$}
		\For{1 $\leq$ \textit{n} $\leq$ \textit{N}}
		\State Randomly select two particles from the swarm, i.e., $\mathbf{Pb}^{k}_{\zeta_1}$, $\mathbf{Pb}^{k}_{\zeta_2}$, and calculate 	$\mathbf{Q}^{k}_{n}$ by (\ref{Q});
		\State Update the velocity and position by (\ref{subswarms1}) and (\ref{QQ});
		\State Evaluate the fitness  of each particle;
		\State $k=k+1$;
		\State Update $\mathbf{Pb}^{k+1}_{n}$ and $\mathbf{G}^{{k+1}}$ by (\ref{personal}) and (\ref{G}), respectively;
		\EndFor
		\EndWhile
		\State \Return  $\mathbf{G}$.
	\end{algorithmic}
\end{algorithm}

\subsection{Convergence Rate Analysis}
{According to the analysis of the PSO convergence rate as indicated
in \cite{trelea2003particle}, it is generally simplified as a 1-D
problem with a single particle. } The update formulas (\ref{subswarms1}) of AUPSO can be simplified and written in the form of a matrix as
\begin{equation} 
	\begin{aligned}
		\begin{bmatrix} 1 & 0 \\ -1 & 1 \end{bmatrix} \begin{bmatrix} {J}_n^{k+1} \\ {R}_n^{k+1} \end{bmatrix} = & \begin{bmatrix} \omega & -c_1 \xi_1 \\ 0 & 1 \end{bmatrix} \begin{bmatrix} {J}_n^{k} \\ {R}_n^{k} \end{bmatrix} + \begin{bmatrix} c_1  \xi_1  {O_n} \\ 0 \end{bmatrix},
	\end{aligned}
\end{equation}
{where $ {O_n} = \left(1-\beta\right){Q_n}+\beta{G}$. Symbols ${J}_n$, $R_n$, $O_n$, $Q_n$, and $G$ represent the elements  of matrices $\mathbf{J}_n$, $\mathbf{R}_n$, $\mathbf{O}_n$, $\mathbf{Q}_n$, and $\mathbf{G}$ at a position $(i,j)$, respectively. For the local exploitation at its later stage, when AUPSO converges, there is a fixed point $X_n^{*} = ({J}_n^{*}, {R}_n^{*}$) that satisfies \cite{cao2018comprehensive}}
\begin{equation} 
	\begin{aligned}
		\begin{bmatrix} 1 & 0 \\ -1 & 1 \end{bmatrix} \begin{bmatrix} {J}_n^{*} \\ {R}_n^{*} \end{bmatrix} = & \begin{bmatrix} \omega & -c_1 \xi_1  \\ 0 & 1 \end{bmatrix} \begin{bmatrix} {J}_n^{*} \\ {R}_n^{*} \end{bmatrix}  + \begin{bmatrix} c_1  \xi_1  {O_n} \\ 0 \end{bmatrix}
	\end{aligned}.
\end{equation}
Define ${X_n} = [{J_n} \quad {R_n}]^{\top}$ and  $\Delta{X_n} = {X}_n^{k} - {X}_n^{*}$. We have 
\begin{equation} 
	\Delta{X}_n^{k+1} = \mathbf{Y} \Delta {X}_n^{k},
\end{equation}
where $\mathbf{Y} = 	\begin{bmatrix} w & -c_1 \xi_1 \\ w & 1-c_1 \xi_1 \end{bmatrix} $.
Then, according to the definition of matrix compatibility norm, we can derive
\begin{equation} 
	\frac{\| \Delta {X}_n^{k+1} \|_1}{\| \Delta {X}_n^{k} \|_1} \leq \| \mathbf{Y} \|_1,
\end{equation}
where $\| \mathbf{Y} \|_1 = \max\{2 w, 1 - 2  c_1  \xi_1\}$. This shows that the norm ratio of particle location deviation $\| \Delta X_n \|$ between two iterations
is a nonzero constant, which indicates that AUPSO has a linear convergence rate, which agrees with the analysis results of
other PSOs \cite{trelea2003particle}. Although it does not have a higher order convergence speed like Newton's method, it provides strong global search capabilities.

\subsection{Computational Complexity of AUPSO} \label{Computational Complexity of AUPSO}
The analysis begins with TPSO. the complexity order of TPSO in (\ref{39}) is
\begin{equation} 
	\mathcal{O}(N(M+1)K_{max}).
\end{equation}
For AUPSO, the complexity order is
\begin{equation} 
	\mathcal{O}(N(M+1)K_{max}).
\end{equation}
Clearly, both TPSO and AUPSO share the same complexity order, ensuring that AUPSO does not introduce additional computational overhead.  

{The proposed AUPSO involves iterative evaluation of the $\operatorname{GDOP_D}$ metric across the ROI and the swarm. Its computational complexity is $\mathcal{O}(N(M+1)K_{max})$, which is linear in the key problem parameters. While this implies a non-negligible computational cost, two critical aspects make it suitable for the considered application in this paper:}

Offline Planning Phase: The optimization is performed once, during the pre-deployment system design phase. The purpose of this paper is to provide the sensors with a robust commanded standby deployment scheme for ROI  localization under drift error.  Therefore, the method's runtime (which may range from seconds to minutes on a standard desktop) does not conflict with the online requirements of target localization.

Practical Acceleration Paths: For very large-scale problems, the computation can be significantly accelerated through parallel evaluation of particles, code optimization, or leveraging GPU computing for the matrix operations inherent in $\operatorname{GDOP_D}$ calculation.

\section{ Simulation Results} \label{Experimental Simulation}
{This section includes three main parts: 1) Validation of $\operatorname{GDOP_D}$ characteristics, which is conducted to support the theoretical findings outlined in Section \ref{Analysis of the characteristics}, 2) GDOP Performance evaluation of the proposed method, and  3) RMSE performance of the proposed method under two drift-error distributions.}

\begin{table*}[!t]
	\caption{{Simulation Parameters for validation of $\operatorname{{GDOP}_{D}}$ Characteristics}}\label{table1:parameters}
	\centering
	\begin{tabular}{ccc}
		\toprule \multicolumn{2}{c}{Parameters} & Values \\
		\midrule
		\multirow{11}{*} {Shared parameters} & Position of sensor 1 ($\mathbf{{s}}_{1}$) & (-30km, -20km, 10km)$^\top$   \\ 
		& Position of  sensor 2 ($\mathbf{{s}}_{2}$) & (50km, -25km, 10km)$^\top$  \\ 
		& Position of sensor 3 ($\mathbf{{s}}_{3}$) & (-10km, 40km, 10km)$^\top$  \\ 
		& Position of sensor 4 ($\mathbf{{s}}_{4}$) & (21.2km, 21.2km, 10km)$^\top$  \\ 
		& \multirow{3}{*}{Region of interest $\mathcal{B}$} & x-direction: [-15km, 10km] \\
		&	& y-direction: [-15km, 10km] \\
		&	& z-direction: 0km \\ 
		&	The size of resolution cells & 50m$\times$50m   \\ 
		&  The light speed $c$ & $3\times10^{8}$m/s\\
		&	\multirow{2}{*}{Position measurement error}  & Mean: 0m \\ 
		&	& Standard deviation ($\sigma _{p}$): 1m \\ 
		&	\multirow{2}{*}{Time-difference measurement error}	 & Mean: 0s \\
		& 	& Standard deviation ($\sigma _{\tau}$): 10ns \\ 
		& Directional correlation coefficients of the drift error of the $m$th sensor &	$\boldsymbol{\bm{\rho}}_{m} = [0, 0,0]^\top$\\	
		\cmidrule(lr){2-3}
		\multirow{2}{*} {Corollary 2 parameters}	& Correlation coefficient of TDOA measurement errors ($\eta$) & 0.5 \\ 
		&{Drift error}	 & $\mu_{d} = 0$, ${\sigma}_{d} = 1 $ \\ 
		\multirow{2}{*} {Corollary 3 parameters}	& Correlation coefficient of TDOA measurement errors ($\eta$) & 0.5 \\ 
		&{Drift error}	 & $\mu_{d} = 1$, ${\sigma}_{d} = 1 $\\ 
		\multirow{2}{*} {Corollary 4 parameters}	& Correlation coefficient of TDOA measurement errors ($\eta$) & 0.95 \\ 
		&	{Drift error}	  & $\mu_{d} = 1$, ${\sigma}_{d} = 3 $\\ 
		\bottomrule
	\end{tabular}
\end{table*}

\begin{table}[h]
	\centering
	\caption{{Application Scenario Parameters}}
	\label{Scenario Parameters}
	\begin{tabular}{cc}
		\toprule
		Parameter & Value \\
		\midrule
		Number of sensors & 8\\ 
		\multirow{3}{*} {Deployment region $\mathcal{A}$} & x-direction: [-5km, 5km] \\
		& y-direction: [-5km, 5km] \\
		& z-direction: [10km, 20km]\\ 
		Region of interest $\mathcal{B}$ & Same as in Table \ref{table1:parameters}\\
		{{Position measurement error}}  & Same as in Table \ref{table1:parameters} \\ 
		{Time-difference measurement error}	 & Same as in  Table \ref{table1:parameters}  \\ 
		{Correlation coefficient ($\eta$)}	 & Same as in Corollary 4 in Table \ref{table1:parameters} \\ \hline
		\multirow{3}{*}{	Sensor 1}  & 	$\mathbf{u}^{dri} _{1} = [1.0, 0.5,0.2]^\top$\\
		& $\mathbf{\bm{\sigma}}^{dri} _{1} = [2.0, 1.5,1.2]^\top$\\
		& $\mathbf{\bm\rho}_{1} = [0.8,0.6,0.2]^\top$\\
		\multirow{3}{*}{	Sensor 2}  & 	$\mathbf{u}^{dri} _{2} = [0.8, 0.3,0.1]^\top$\\
		& $\mathbf{\bm\sigma}^{dri} _{2} = [1.8, 2.3,3.1]^\top$\\
		& $\mathbf{\bm\rho}_{2} = [0.5, 0.2, 0.6]^\top$\\
		\multirow{3}{*}{	Sensor 3}  & 	$\mathbf{u}^{dri} _{3} = [1.2, 0.7, 0.3]^\top$\\
		& $\mathbf{\bm\sigma}^{dri} _{3} = [1.2, 0.7, 0.3]^\top$\\
		& $\mathbf{\bm\rho}_{3} = [0.6, 0.8, 0.27]^\top$\\
		\multirow{3}{*}{	Sensor 4}  & 	$\mathbf{u}^{dri} _{4} = [1.05, 0.1, 1.5]^\top$\\
		& $\mathbf{\bm\sigma}^{dri} _{4} = [0.7, 1.5, 2.3]^\top$\\
		& $\mathbf{\bm\rho}_{4} = [0.4, 0.5, 0.1]^\top$\\
		\multirow{3}{*}{	Sensor 5}  & 	$\mathbf{u}^{dri} _{5} = [2.1, 0.9,1.8]^\top$\\
		& $\mathbf{\bm\sigma}^{dri} _{5} = [1.5,2.9,3.1]^\top$\\
		& $\mathbf{\bm\rho}_{5} = [0.7,0.5,0.65]^\top$\\
		\multirow{3}{*}{	Sensor 6}  & 	$\mathbf{u}^{dri} _{6} = [2.1, 1.8,2.5]^\top$\\
		& $\mathbf{\bm\sigma}^{dri} _{6} = [2.5,1.8,0.9]^\top$\\
		& $\mathbf{\bm\rho}_{6} = [0.9,0.43,0.75]^\top$\\
		\multirow{3}{*}{	Sensor 7}  & 	$\mathbf{u}^{dri} _{7} = [4.2,1.8,3.2]^\top$\\
		& $\mathbf{\bm\sigma}^{dri} _{7} = [1.5,2.8,3.6]^\top$\\
		& $\mathbf{\bm\rho}_{7} = [0.75,0.21,0.64]^\top$\\
		\multirow{3}{*}{	Sensor 8}  & 	$\mathbf{u}^{dri} _{8} = [2.1,0.9,1.5]^\top$\\
		& $\mathbf{\bm\sigma}^{dri} _{8} = [2.1,3.1,1.8]^\top$\\
		& $\mathbf{\bm\rho}_{8} = [0.13,0.24,0.35]^\top$\\
		\bottomrule
	\end{tabular}
\end{table}

\begin{table}[h]
	\centering
	\caption{Parameters of AUPSO}
	\label{Algorithmic Parameters}
	\begin{tabular}{cc}
		\toprule
		Parameter& Value \\
		\midrule
		Inertia weight ($w$) & 0.9-0.4 \cite{kennedy1995particle}  \\ 
		Cognitive acceleration coefficient ($c_1$) & 1.49455 \cite{kennedy1995particle}  \\
		\bottomrule
	\end{tabular}
\end{table}

\begin{table}[h]
	\caption{Statistical Results of GDOP under $\operatorname{GDOP^{max}_{T}}$ and $\operatorname{GDOP^{max}_{D}}$}
	\label{Statistical Results of GDOP under}
	\centering
	\begin{tabular}{ccc}
		\toprule
		~   & Max.$\operatorname{GDOP^{max}_{D}}$    & Ave.$\operatorname{GDOP^{max}_{D}}$   \\ 
		\midrule
		Optimized by $\operatorname{GDOP_{T}}$   & 1.59E+06 &  2.36E+05 \\ 
		Optimized by $\operatorname{GDOP_{D}}$ & \textbf{7.77E+03} &  \textbf{3.22E+02}\\ 
		\bottomrule
	\end{tabular}
\end{table}

\begin{table}[h]
	\caption{{Statistical Results of $\operatorname{GDOP^{max}_{D}}$ for Different methods}}
	\label{Statistical Results of GDOP}
	\centering
	\scriptsize
	\begin{tabular}{cccc}
		\toprule
		~   & Max.$\operatorname{GDOP^{max}_{D}}$    & Ave.$\operatorname{GDOP^{max}_{D}}$ &  Std.$\operatorname{GDOP^{max}_{D}}$  \\
		\midrule
		TPSO \cite{kennedy1995particle}  & 2.89E+04 &  2.73E+03 & 6.42E+03 \\ 
		AWPSO \cite{liu2019novel} & 1.97E+04 &  9.74E+02 & 3.12E+03 \\ 
		CHpPSO-ABS \cite{zhang2026complementary} &  5.10E+04 &  2.98E+03 & 7.92E+03\\
		\textbf{AUPSO} & \textbf{7.77E+03} &  \textbf{3.22E+02} & \textbf{1.08E+03} \\ 
		Improvement	ratio &\textbf{{153.55$\%$}} & \textbf{202.39$\%$} & \textbf{188.90$\%$}\\
		\bottomrule
	\end{tabular}
\end{table}

\begin{table}[h]
	\caption{{The average runtime  of different methods.}}
	\label{The average runtime}
	\centering
	\scriptsize
	\begin{tabular}{ccccc}
		\toprule
		~   & TPSO \cite{kennedy1995particle}    & 	AWPSO \cite{liu2019novel} &  	CHpPSO-ABS \cite{zhang2026complementary} & 	\textbf{AUPSO} \\
		\midrule
		Runtime (s)  & 2.10E+01 &  2.09E+01 & 2.09E+01 & 2.08E+01  \\ 
		\bottomrule
	\end{tabular}
\end{table}

\subsubsection{ {RMSE Performance of the Proposed Method Compared With Baseline Methods}}
In this subsection, we introduce uniform, LHS, SQP \cite{costa2025aggregation} methods alongside TPSO, AWPSO, and CHpPSO-ABS. The uniform method places eight sensors at the eight vertices of the deployment region $\mathcal{A}$. Moreover, SQP, TPSO, AWPSO, CHpPSO-ABS, and AUPSO equally generate $NK_{max}$ potential sensor deployment schemes in each Monte Carlo trial. To ensure a fair comparison, the LHS method similarly produces $NK_{max}$ schemes in each trial, collecting the deployment scheme corresponding to the minimum $\operatorname{GDOP^{max}_{D}}$. The RMSE performance corresponding to the sensor deployment schemes generated by all methods , with performance measured by the median of the $\operatorname{RMSE^{max}}$  (Med.$\operatorname{RMSE^{max}}$, reflecting worst-case localization error) and the median of the $\operatorname{RMSE^{mean}}$ (Med.$\operatorname{RMSE^{mean}}$, reflecting average localization accuracy) across the ROI.
	
The results illustrate the performance hierarchy among baselines. While  Uniform, LHS and SQP methods result in significantly higher localization errors, the PSO variants (TPSO, AWPSO, CHpPSO-ABS) perform considerably better. Among them, AWPSO generally ranks the second best, though its performance is still significantly and meaningfully outperformed by AUPSO. Besides, under both Gaussian and uniform drift-error distributions, the proposed AUPSO method consistently achieves the lowest values for both RMSE metrics, confirming its robust superiority. Specifically, for Med.$\operatorname{RMSE^{max}}$, AUPSO demonstrates an improvement ratio of $14.75\%$ over the best baseline (AWPSO) under a Gaussian distribution and a notable 19.77$\%$ under a uniform distribution. For Med.$\operatorname{RMSE^{mean}}$, the improvement over the best baseline (AWPSO) remains consistent at 10.66$\%$ for both distributions. The profound superiority of AUPSO is rigorously validated by the Wilcoxon rank-sum test. All reported $p$-values comparing AUPSO to each baseline are exceedingly small (e.g., on the order of 
$10^{-2}$  to $10^{-34}$), far below the significance level of 0.05. This provides overwhelming statistical evidence of the performance gains of AUPSO and shows that AUPSO generates sensor deployment schemes that yield significantly lower localization errors compared to all baseline methods. The simulation results also show that the gains are statistically robust to  error distributions, and quantitatively substantial, underscoring the effectiveness of the proposed method under drift-error environments.

\section{Conclusion} \label{Conclusion}
In this paper, we investigated the sensor deployment optimization problem for passive TDOA localization across the entire ROI under drift error. To tackle this challenge, we firstly introduced a generalized drift-aware $\operatorname{{GDOP}_{D}}$ metric, which is a generalization of the traditional $\operatorname{{GDOP}_{T}}$ metric. We then theoretically investigated the mathematical relationship between $\operatorname{GDOP_D}$ and $\operatorname{GDOP_T}$, and theoretically revealed the mechanism of how drift errors affect GDOP. We then formulated a non-convex min-max sensor deployment optimization model aimed at minimizing the worst-case $\operatorname{{GDOP}_{D}}$ across the ROI. Next, to solve this challenging problem, we proposed an AUPSO method, which alleviates the premature convergence and the oscillatory behavior of the traditional PSO. Extensive simulations validated the theoretical results and demonstrated the superiority of the proposed method in terms of localization accuracy.   

Based on this study, future work can focus on developing a  method to dynamically adjust sensor positions online to compensate for real-time drift errors.  Specifically, considering the time-varying attitude dynamics, we plan to formulate an online and real-time UAV attitude compensation and active trajectory adjustment framework. By coupling the offline robust topology baseline established in this paper with an online active attitude compensation mechanism, we aim to develop a fully closed-loop and high-precision UAV TDOA localization system under severe dynamic disturbances. In addition, extending the proposed deployment optimization framework to non-line-of-sight (NLOS) environments constitutes another important direction. Future work could incorporate terrain or building maps into the offline planning stage, enabling the optimization to favor sensor geometries that maintain redundant line-of-sight coverage in NLOS-prone sub-regions. Such an extension would integrate NLOS awareness into the geometric design and further enhance the robustness of the overall localization system.

\bibliographystyle{IEEEtran}
\bibliography{Main_Manuscript}
\vfill

\newpage

\end{document}